\RequirePackage{xcolor}
\documentclass[journal,twoside,web]{ieeecolor}
\usepackage{generic}

\usepackage{subcaption}
\usepackage{cite}

\usepackage{amssymb,amsmath,latexsym,amsfonts,amsthm,mathtools}

\usepackage[colorlinks=true]{hyperref}
\hypersetup{colorlinks, breaklinks, citecolor=blue, linkcolor=blue, urlcolor=blue}
\usepackage{graphicx}
\usepackage{float}
\usepackage{textcomp}
\definecolor{darkpastelpurple}{rgb}{0.59, 0.44, 0.84}
\usepackage[nameinlink]{cleveref}
\Crefformat{figure}{#2Fig.~#1#3}
\Crefmultiformat{figure}{Figs.~#2#1#3}{ and~#2#1#3}{, #2#1#3}{ and~#2#1#3}

\usepackage{tikz}
\usetikzlibrary{automata, shapes, arrows, calc, arrows.meta, fit, positioning}

\theoremstyle{plain}
\newtheorem{theorem}{Theorem}
\newtheorem{lemma}{Lemma}

\newtheorem*{problem*}{Problem}
\theoremstyle{remark}
\newtheorem{remark}{Remark}
\newtheorem{assumption}{Assumption}
\crefname{assumption}{assumption}{assumptions}
\Crefname{assumption}{Assumption}{Assumptions}

\newtheorem{definition}{Definition}
\theoremstyle{definition}

\newcommand{\sign}{\mathrm{sign}}
\usepackage{mathrsfs}

\usepackage{newtxmath}
\usepackage{booktabs}
\usepackage{duckuments}

\begin{document}

	\title{Predefined-Time Leaderless Consensus Under Denial-of-Service Attacks}
	\author{Lohitvel Gopikannan,~\IEEEmembership{Student Member, IEEE}, Shashi Ranjan Kumar,~\IEEEmembership{Senior Member, IEEE}, \newline and Abhinav Sinha,~\IEEEmembership{Senior Member,~IEEE}
		\thanks{L. Gopikannan and S. R. Kumar are with the Intelligent Systems \& Control (ISaC) Lab, Department of Aerospace Engineering, Indian Institute of Technology Bombay, Mumbai 400076, India. (e-mails: lohitvel@aero.iitb.ac.in, srk@aero.iitb.ac.in). A. Sinha is with the Guidance, Autonomy, Learning, and Control for Intelligent Systems (GALACxIS) Lab, Department of Aerospace Engineering, University of Cincinnati, OH 45221, USA. (e-mail: abhinav.sinha@uc.edu).}
	}
	\maketitle

	\begin{abstract}
    This paper addresses predefined-time resilient consensus of leaderless second-order nonlinear multi-agent systems under denial-of-service (DoS) attacks, motivated by coordination requirements in safety-critical applications. The agents are subject to bounded external disturbances and communicate over a strongly connected directed graph whose links are simultaneously disabled during attacks. We develop a switching sliding-mode protocol with the objective of reaching an invariant manifold of position and velocity agreement. The protocol uses relative position and velocity information during attack-free intervals and local velocity feedback during communication blackouts. A time-scaling function remains constant during each blackout and resumes evolving when communication is restored, accounting for the time available for consensus. Under bounds on attack duration and frequency, we derive sufficient gain conditions through a Lyapunov analysis. We show that, despite bounded disturbances, the agents achieve position and velocity consensus by a realistic settling time equal to a prescribed convergence duration plus the cumulative attack duration up to the realistic settling time. The prescribed convergence duration is independent of the initial conditions, and the realistic settling time reduces to that duration in the absence of attacks.
	\end{abstract}
	
\section{Introduction}\label{sec:introduction}
Coordinated motion planning, e.g., \cite{9448334,9924233,9274339}, relies on information exchanged among agents to organize their motion. In tasks with a scheduled cooperative maneuver, e.g. \cite{9000526}, the time required to reach agreement becomes part of the coordination objective. For second-order agents, both position and velocity disagreement matter because velocity mismatch continues to change relative positions. A denial-of-service (DoS) attack interrupts the information exchange needed to correct those mismatches while the agents continue to evolve. The interaction between physical motion, communication availability, and convergence time motivates resilient consensus for safety-critical multi-agent systems (MASs).

For directed second-order networks, the work in \cite{yu2009second} established consensus conditions that account for nonlinear dynamics and position--velocity coupling. Asymptotic agreement provides no finite completion deadline, which motivated finite-time protocols for double-integrator networks, including the leaderless and leader-following designs of \cite{li2011finite}. A finite-time bound can still depend on the initial disagreement. Fixed-time protocols for nonlinear and disturbed second-order agents~\cite{hong2018novel} address this dependence through settling-time bounds that are uniform over initial conditions. The ensuing design objective is to make a desired convergence duration a transparent controller parameter, which has become a central motivation for prescribed-time consensus~\cite{ning2022fixed}.

The authors in \cite{wang2018prescribed} developed time-scaling methods that prescribe the convergence time for single-integrator consensus and containment problems. Extending such guarantees to second-order motion requires control of both position and velocity. This requirement was addressed in \cite{ding2023prescribed} through prescribed-time observers and formation-tracking controllers for second-order networks with a dynamic leader. The leader trajectory yields the tracking objective, whereas leaderless coordination requires the agents to reach agreement through their mutual interactions. The work in \cite{meng2026practical} advanced prescribed-time coordination to practical scaled consensus in disturbed second-order networks without any cyber-attacks, including a leaderless setting, with estimates of the residual consensus error. 

Sliding-mode designs, albeit robust, must also accommodate the structure of the leaderless disagreement dynamics. Therein graph Laplacian is singular, which precludes using its inverse in the consensus protocol. Integral sliding-mode designs for disturbed second-order agents have shown to achieve asymptotic leaderless consensus and finite-time leader-following consensus~\cite{wang2018sliding}. Nonsingular sliding-mode control with coded event-triggered communication has subsequently yielded predefined-time leaderless consensus~\cite{wei2025further}.

Resilient consensus methods have addressed DoS-induced communication loss in several settings. For example, event-triggered control for nonlinear second-order agents under DoS attacks over directed interaction topologies was developed in \cite{shang2021event,liang2025fixed}. An asymptotic leaderless consensus using event- and self-triggered communication when individual channels are independently attacked was established in \cite{li2026event}. Dual-terminal triggering also supports resilient leaderless average consensus while reducing communication and actuator updates~\cite{li2026dual}. These studies address the availability and use of network information. A scheduled coordination task additionally requires a finite convergence-time guarantee. For second-order agents, that guarantee must accommodate the evolution of velocity and position disagreement throughout each blackout. To this end, the authors in \cite{ye2026prescribed} studied prescribed-time consensus under DoS by introducing a realistic settling time that adds the cumulative attack duration to a prescribed convergence duration. Their leader-following formulation accounts for the communication time lost to attacks. A prescribed-time event-triggered control for leader-following nonlinear second-order systems with time delays, DoS attacks, and external disturbances was presented in the work of \cite{zhu2026prescribed}. Prescribed-time formation under DoS has also been developed for leaderless double-integrator networks with undirected communication graphs and self-triggered updates~\cite{sharma2026selftriggered}.

To the best of the authors' knowledge, sliding-mode coordination laws that achieve leaderless consensus within a prescribed cumulative duration of attack-free communication remain comparatively scarce. The central difficulty is that velocity mismatch and external disturbances continue to drive relative motion during each blackout, while the neighboring-agent information needed to correct the disagreement is unavailable. A controller must govern the agents using locally available information during attacks and guarantee convergence by a realistic settling time that accounts for the cumulative loss of communication.

We address this problem through a sliding-mode framework over a strongly connected directed graph. Our main contributions are as follows. First, we design a switching sliding-mode protocol for predefined-time leaderless consensus of second-order agents with known nonlinear dynamics and bounded disturbances. The position--velocity sliding variable is built from mutual disagreement, and the protocol avoids inversion of the singular graph Laplacian. Local velocity feedback governs the agents during complete communication blackouts. Second, We extend the realistic-settling-time concept of~\cite{ye2026prescribed} to prescribe the cumulative attack-free communication time for leaderless consensus independently of the initial conditions, with the convergence deadline extended by the accumulated outage duration. The extension realizes this timing guarantee through a sliding variable built from mutual position and velocity disagreement, accommodating continued second-order motion under local feedback while the associated time scaling pauses during blackouts. Third, we derive a quantitative disagreement bound that tracks amplification and disturbance accumulation during successive blackouts and their attenuation after communication recovery. The bound relates the convergence rate to the directed graph and control gains, establishing position and velocity agreement while keeping the consensus correction bounded as the time scaling becomes singular at the realistic settling time.

\section{Background and Problem Formulation}
In this work, the communication topology is modeled by a weighted directed graph $\mathcal{G} = (\mathcal{V}, \mathcal{E})$, where $N\geq2$ and $\mathcal{V} = \{1, \ldots, N\}$ is the agent set and $(i,j) \in \mathcal{E} \subseteq \mathcal{V} \times \mathcal{V}$ indicates that agent $j$ receives information from agent $i$. The weighted adjacency matrix $\mathcal{A} = [a_{ij}] \in \mathbb{R}^{N\times N}$ has $a_{ij} > 0$ if $(j,i) \in \mathcal{E}$ and zero otherwise, with $a_{ii} = 0$. The in-degree matrix $\mathcal{D} = \operatorname{diag}\{d_1, \ldots, d_N\}$ has $d_i = \sum_{j=1}^N a_{ij}$, and the graph Laplacian is defined as $\mathcal{L} = \mathcal{D} - \mathcal{A} \in \mathbb{R}^{N\times N}$. $\mathcal{L}$ has a simple zero eigenvalue with all other eigenvalues having positive real parts if and only if the directed network has a directed spanning tree.
\begin{assumption}\label{asm:strong}
    $\mathcal{G}$ is assumed to be strongly connected.
\end{assumption}
\begin{lemma}[\cite{yu2009second}]\label{lem:directed_graph}
Under \Cref{asm:strong}, $\mathcal{L}\mathbf{1}_N=\mathbf{0}_N$, and the Laplacian admits a unique normalized vector $\boldsymbol{\xi}=(\xi_1,\ldots,\xi_N)^\top\in\mathbb{R}_{>0}^N$ satisfying $    \boldsymbol{\xi}^\top\mathcal{L}=\mathbf{0}_N^\top,~
    \boldsymbol{\xi}^\top\mathbf{1}_N=1.$ Define $\Xi=\operatorname{diag}(\xi_1,\ldots,\xi_N)$ and $\hat{\mathcal{L}}=(\Xi\mathcal{L}+\mathcal{L}^\top\Xi)/2$.
The matrix $\hat{\mathcal{L}}$ is symmetric positive semidefinite with $\ker(\mathcal{L})=\ker(\hat{\mathcal{L}})=\operatorname{span}\{\mathbf{1}_N\}$.
The graph constant $    a(\hat{\mathcal{L}}):=
    \min_{\substack{\boldsymbol{\xi}^\top\mathbf{z}=0\\
    \mathbf{z}\in\mathbb{R}^N\setminus\{\mathbf{0}_N\}}}
    \dfrac{\mathbf{z}^\top\hat{\mathcal{L}}\mathbf{z}}{\|\mathbf{z}\|_2^2}$ satisfies $0<a(\hat{\mathcal{L}})\leq\lambda_2(\hat{\mathcal{L}})$, where $\lambda_2(\hat{\mathcal{L}})$ is the second-smallest eigenvalue of $\hat{\mathcal{L}}$.
\end{lemma}
The agents associated with each node of $\mathcal{G}$ are described by the following dynamics:
\begin{align}
\begin{cases}
\dot{x}_i = v_i,\\
\dot{v}_i = f_i + b_i u_i + d_i,
\end{cases}
\quad i = 1, \ldots, N,
\label{eq:agent_dynamics}
\end{align}
where, for the $i$-th agent, $x_i, v_i \in \mathbb{R}^n$ denote the position and velocity, $f_i \coloneqq f(x_i, v_i, t)\in\mathbb{R}^n$ and $b_i \coloneqq b(x_i, v_i, t) \in \mathbb{R}^{n\times n}\neq 0$ are known nonlinear functions representing the inherent dynamics and state-dependent input matrix, respectively, with $u_i$ as the control input to be designed, and $d_i(t)$ denotes the external bounded disturbance. For simplicity, scalar second-order dynamics ($n=1$) are considered throughout. The results can be extended to higher-dimensional systems by employing the Kronecker product.

A DoS attack occurs when an adversary blocks or disrupts the communication
links between agents, preventing the exchange of information over the
affected connections. Let $t_0$ be the initial time, and let $t_m^s$, 
$t_m^e$, $m=1,2,\ldots$, denote the start and end instants of the
$m^\text{th}$ attack. The attack intervals are indexed chronologically, with $t_0^e=t_0\leq t_1^s$ and $t_m^s<t_m^e\leq t_{m+1}^s$. The attacked and attack-free portions of $[t_0,t]$ are
written as $\Upsilon_A(t_0,t) := [t_0,t]\cap\bigcup_{j\geq1}[t_j^s,t_j^e)$
and $\Upsilon_N(t_0,t) := [t_0,t] \setminus \Upsilon_A(t_0,t)$, with
$|\Upsilon_A(t_0,t)| + |\Upsilon_N(t_0,t)| = t - t_0$. For
$t \in [t_{m-1}^e, t_m^s)$, $|\Upsilon_A(t_0,t)| = \sum_{j=1}^{m-1}(t_j^e - t_j^s)$,
whereas for $t \in [t_m^s, t_m^e)$, $|\Upsilon_A(t_0,t)| = \sum_{j=1}^{m-1}(t_j^e - t_j^s) + (t - t_m^s)$. Let $\mathcal{N}_A(t_0,t)$ denote the number of attack onsets in $[t_0,t]$ and they are detectable by the agents.  In this work,
zero-topology attack is considered, whose characteristics are described below.
\begin{definition}[\cite{cao2026distributed}]
Under a zero-topology attack, all communication edges are
disabled during each attack interval and restored during attack-free intervals.
The availability function satisfies $\zeta(t)=0$ during attacks and $\zeta(t)=1$ during attack-free intervals,
so that the time-varying adjacency matrix under DoS attacks is given by
$\mathcal{A}^{\zeta}(t) = [a_{ij}^{\zeta}(t)]$, where
$a_{ij}^{\zeta}(t) = a_{ij}\,\zeta(t)$. 
\end{definition}
\begin{assumption}\label{assum:duration}
There exist constants $\nu_1,\nu_3 > 1$ and $\nu_2,\nu_4 \ge 0$ such that
\begin{equation}
 |\Upsilon_A(t_0,t)| \le \frac{t - t_0}{\nu_1} + \nu_2,\,\,\mathcal{N}_A(t_0,t) \le \frac{t - t_0}{\nu_3} + \nu_4.
\label{eq:assumption_duration}
\end{equation}
\end{assumption}
\begin{remark}\label{remark:energy}
Jamming a communication channel requires the attacker to continuously
expend energy, and since attackers generally operate under a finite energy
budget, attacks cannot persist indefinitely. This physical
constraint justifies bounding the duration and frequency of such attacks \cite{cao2026distributed,ye2026prescribed,li2026event}.
\end{remark}
A fixed convergence deadline requires sufficient communication time and control authority to remove the remaining disagreement.
Online gain tuning can accelerate correction after communication resumes, but the duration and frequency bounds in \Cref{assum:duration} do not specify future recovery instants.
During a complete blackout, increasing relative-state feedback gains cannot restore the missing neighboring-agent information.
A late recovery also compresses the correction into a short interval, i.e., even closing a fixed position mismatch from rest requires peak acceleration that grows inversely with the square of the available recovery time. With finite actuation, this requirement can make the original deadline infeasible.
We adopt the realistic settling time of~\cite{ye2026prescribed} to account for lost communication time while retaining a prescribed duration of attack-free coordination.
\begin{definition}[\cite{ye2026prescribed}]
Let $l$ and $n_u$ be positive integers denoting the state and input dimensions, respectively.
Consider the non-autonomous system $\dot{\psi}(t) = h(\psi,u,t)$, where
$\psi(t) \in \mathbb{R}^l$ is the state vector, $u \in \mathbb{R}^{n_u}$ is the
control vector, and $h : \mathbb{R}^l \times \mathbb{R}^{n_u} \times
\mathbb{R}_+ \to \mathbb{R}^l$ is a nonlinear function satisfying
$h(0,0,t) = 0$. Then the  origin of the system is said to be
globally uniformly resiliently prescribed-time stable if it is globally
uniformly finite-time stable, and the settling time $T(\psi_0)$ admits a finite bound
$t_c^F \in \mathbb{R}_+$ satisfying $T(\psi_0) \le t_c^F$ for all
$\psi_0 \in \mathbb{R}^l$. The bound $t_c^F$, termed the realistic settling time, satisfies $t_c^F = t_c + |\Upsilon_A(0,t_c^F)|$,
with $t_c$ stipulated in advance, independent of the initial conditions
and of the system/controller parameters. Furthermore, there exists a radially unbounded positive-definite function $V : \mathbb{R}^l \times \mathbb{R}_{\geq 0} \to \mathbb{R}_+$ satisfying $ V(\psi, t_m^s) \leq \alpha V(\psi, t_m^{s-})$, $ V(\psi,t_m^e) \leq \alpha V(\psi, t_m^{e-})$ and $\dot{V}(\psi, t) \leq \theta V(\psi, t), \forall\, t \in [0, t_c^F) \cap [t_m^s, t_m^e)$ in the absence of disturbances, where $\alpha \geq 1$ and $\theta > 0$.
\label{def:resilient_prescribed_time}
\end{definition}
\begin{problem*}%[Resilient second-order leaderless consensus]
Set $t_0=0$ and consider the scalar agents \eqref{eq:agent_dynamics} over a graph satisfying \Cref{asm:strong}, subject to zero-topology attacks satisfying \Cref{assum:duration}.
Given a prescribed duration $t_c>0$, chosen independently of the initial states, and a known disturbance bound $\bar{d}\geq0$, design causal switching control inputs $u_i(t)$, $i=1,\ldots,N$, using local measurements and available network information during attack-free intervals, and local state feedback during blackouts.
For each admissible attack history, let us define the realistic settling time as the first instant at which the cumulative attack-free time reaches $t_c$, i.e.,
\begin{equation}\label{eq:realistic_settling_time}
    t_c^F := \inf\left\{t\geq0:\,|\Upsilon_N(0,t)|\geq t_c\right\}.
\end{equation}
For all finite initial states $(x_i(0),v_i(0))\in\mathbb{R}^2$ and all measurable disturbance signals satisfying $|d_i(t)|\leq\bar{d}$ for all agents and $t\geq0$, we require the closed-loop trajectories to satisfy, for every pair $i,j\in\{1,\ldots,N\}$,
\begin{align}
    \lim_{t\to(t_c^F)^-}\left(x_i(t)-x_j(t)\right) =& 0,~
    \lim_{t\to(t_c^F)^-}\left(v_i(t)-v_j(t)\right) = 0.
    \label{eq:consensus}
\end{align}
\end{problem*}
The cumulative attack-free time is continuous, and \Cref{assum:duration} yields $|\Upsilon_N(0,t)|\geq(1-1/\nu_1)t-\nu_2$.
The deadline in \eqref{eq:realistic_settling_time} consequently satisfies
\begin{equation}\label{eq:realistic_settling_bound}
    t_c^F = t_c + |\Upsilon_A(0,t_c^F)|
    \leq \frac{\nu_1(t_c+\nu_2)}{\nu_1-1}.
\end{equation}
The bound in \eqref{eq:realistic_settling_bound} is finite and independent of the initial states.
During a blackout, the cumulative attack-free time remains fixed while relative motion and disturbances continue to change the disagreements.
The controller must regulate the blackout motion through local feedback and use restored communication to drive both disagreements to zero.

\section{Main Results}
We formulate the objective \eqref{eq:consensus} in terms of position and velocity disagreements induced by $\mathcal{L}$. Let $\mathbf{x} = [x_1, \dots, x_N]^\top$ and $\mathbf{v} = [v_1, \dots, v_N]^\top$ denote the global position and velocity state vectors, respectively. The system \eqref{eq:agent_dynamics} can then be written in compact form
\begin{align}
    \dot{\mathbf{x}} &= \mathbf{v}, \\
    \dot{\mathbf{v}} &= \mathbf{F} + \mathbf{B}\mathbf{u}+\mathbf{d},
\end{align}
where $\mathbf{F} = [f_1, \dots, f_N]^\top \in \mathbb{R}^N$, $\mathbf{B} = \operatorname{diag}(b_1, \dots, b_N) \in \mathbb{R}^{N \times N}$, $\mathbf{d} = [d_1, \dots, d_N]^\top \in \mathbb{R}^N$ and $\mathbf{u} = [u_1, \dots, u_N]^\top \in \mathbb{R}^N$. Define the consensus error vectors $\mathbf{Z}_1 = \mathcal{L}\mathbf{x} \in \mathbb{R}^N$ and $\mathbf{Z}_2 = \mathcal{L}\mathbf{v} \in \mathbb{R}^N$, whose $i^\text{th}$ components are $z_{1_i} = \sum_{j=1}^N a_{ij}(x_i - x_j)$ and $z_{2_i} = \sum_{j=1}^N a_{ij}(v_i - v_j)$, representing the local disagreement of agent $i$ in position and velocity, respectively.
It can be noted that $\mathbf{Z}_1 = \mathbf{Z}_2 = \mathbf{0}$ implies consensus among all agents. 

To facilitate the control design, the following sliding surface is introduced:
\begin{equation}\label{eq:sliding_variable}
    \mathbf{S} = \mathbf{Z}_2 + \gamma(t)\mathbf{Z}_1,
\end{equation}
where the time-varying scaling function~\cite{ye2026prescribed} is
\begin{equation}\label{eq:time_scaling}
    \gamma(t) = \begin{cases}
        \dfrac{c}{t_c+|\Upsilon_A(0,t)| - t}, & 0 \leq t < t_c^F, \\
        0, & t \geq t_c^F,
    \end{cases}
\end{equation}
with $c \geq 2$. 
\begin{remark}
The scaling function $\gamma(t)$ is not unique and may be chosen by the
user; several such choices exist in the literature; e.g., see \cite{yang2026finite}. In general,
for $T>0$, one may define $s(t)$ with $s(0)=0$,
$s(T)=+\infty$, strictly monotonically increasing and continuously
differentiable, giving the scaling function $\gamma(T,t) := \dot{s}(t)$, which
similarly satisfies $\gamma(T,t) \to \infty$ as $t \to T$. Although this
singularity appears restrictive, remedies have been proposed in the
literature \cite{yang2026finite}. As this work focuses on consensus protocol design
rather than the scaling function itself, we leave the choice of
$\gamma(t)$ to the user.
\end{remark}
\begin{theorem}\label{thm:consensus}
Consider the scalar agents \eqref{eq:agent_dynamics} under \Cref{asm:strong,assum:duration}, with $|d_i(t)|\leq\bar{d}$.
For a prescribed $t_c>0$, choose gains satisfying
\begin{align}
    k >& \frac{\lambda_{\max}(\Xi)}{a(\hat{\mathcal{L}})}\left(1+\frac{1}{c}\right),~~
    c \geq 2,~k_v\geq0,~\eta\geq\bar{d}.
    \label{eq:consensus_gains}
\end{align}
With $\mathbf{S}$ and $\gamma(t)$ defined by \eqref{eq:sliding_variable} and \eqref{eq:time_scaling}, apply the switching protocol
\begin{equation}\label{eq:control}
    \mathbf{u}(t)=-\mathbf{B}^{-1}
    \begin{cases}
       \mathbf{F}+k\gamma(t)\mathbf{S}
       +\eta\sign(\mathcal{L}^\top\Xi\mathbf{S}); & \zeta(t)=1,\\
       \mathbf{F}+k_v\mathbf{v}+\eta\sign(\mathbf{v}); & \zeta(t)=0.
    \end{cases}
\end{equation}
Then the Filippov solutions satisfy \eqref{eq:consensus} for all finite initial states and admissible attack histories, with $t_c^F$ defined by \eqref{eq:realistic_settling_time}.
The attack-free consensus correction $k\gamma(t)\mathbf{S}(t)$ remains bounded as $t\to(t_c^F)^-$ and converges to zero when $c>2$.
\end{theorem}
\begin{proof}
All differential inequalities below hold almost everywhere along Filippov solutions. On an attack-free interval $[t_m^e,t_{m+1}^s)\cap[0,t_c^F)$, define $t_c^{m+1}:=t_c+|\Upsilon_A(0,t_m^e)|$, with $t_0^e=0$. The scaling law \eqref{eq:time_scaling} satisfies $\gamma(t)=c/(t_c^{m+1}-t)$ on this interval. Differentiating \eqref{eq:sliding_variable} and substituting \eqref{eq:control} yields
\begin{align}\label{eq:S_dot}
    \dot{\mathbf{S}} =& -\gamma(k\mathcal{L}-\mathbf{I}_N)\mathbf{S}-\mu\mathbf{Z}_1+\mathcal{L}\left(\mathbf{d}-\eta\sign(\mathcal{L}^{\top}\Xi\mathbf{S})\right),
\end{align}
where $\mu:=\gamma^2-\dot{\gamma}=c(c-1)/(t_c^{m+1}-t)^2>0$ by $c\geq2$ in \eqref{eq:consensus_gains}. Consider the piecewise differentiable Lyapunov function
\begin{equation}\label{eq:consensus_lyapunov}
    V(t)=\frac{1}{2}\mathbf{S}^{\top}\Xi\mathbf{S}
    +\frac{1}{2}\mu(t)\mathbf{Z}_1^{\top}\Xi\mathbf{Z}_1.
\end{equation}
Using $\dot{\mathbf{Z}}_1=\mathbf{Z}_2=\mathbf{S}-\gamma\mathbf{Z}_1$ and \eqref{eq:S_dot}, the mixed terms in $\dot{V}$ cancel, leaving
\begin{align}\label{eq:Vdot_pre1}
    \dot{V} =& -\gamma\mathbf{S}^{\top}\Xi(k\mathcal{L}-\mathbf{I}_N)\mathbf{S}
    -\hat{\mu}\mathbf{Z}_1^{\top}\Xi\mathbf{Z}_1\nonumber\\
    &+\mathbf{S}^{\top}\Xi\mathcal{L}\left(\mathbf{d}-\eta\sign(\mathcal{L}^{\top}\Xi\mathbf{S})\right),
\end{align}
where $\hat{\mu}:=\mu\gamma-\dot{\mu}/2=(c-1)\mu\gamma/c>0$. H\"older's inequality and the uniform disturbance bound imply
$\mathbf{S}^{\top}\Xi\mathcal{L}(\mathbf{d}-\eta\sign(\mathcal{L}^{\top}\Xi\mathbf{S}))\leq(\bar{d}-\eta)\|\mathcal{L}^{\top}\Xi\mathbf{S}\|_1\leq0$ by \eqref{eq:consensus_gains}.
Moreover, $\boldsymbol{\xi}^{\top}\mathbf{S}=0$ because $\mathbf{Z}_1=\mathcal{L}\mathbf{x}$ and $\mathbf{Z}_2=\mathcal{L}\mathbf{v}$. Applying \Cref{lem:directed_graph} to \eqref{eq:Vdot_pre1} yields
\begin{equation}\label{eq:V_graph_decay}
    \dot{V}\leq-\gamma\hat{k}\mathbf{S}^{\top}\Xi\mathbf{S}
    -\frac{c-1}{c}\mu\gamma\mathbf{Z}_1^{\top}\Xi\mathbf{Z}_1,
\end{equation}
where $\hat{k}:=[k a(\hat{\mathcal{L}})-\lambda_{\max}(\Xi)]/\lambda_{\max}(\Xi)>1/c$ by \eqref{eq:consensus_gains}.
Set $\bar{c}:=2c\min\{\hat{k},(c-1)/c\}\geq2$. From \eqref{eq:consensus_lyapunov} and \eqref{eq:V_graph_decay}, one obtains $\dot{V}\leq-\bar{c}V/(t_c^{m+1}-t)$. Integration over the attack-free interval establishes
\begin{equation}\label{eq:V_healthy1}
    V(t)\leq\left(\frac{t_c^{m+1}-t}{t_c^{m+1}-t_m^e}\right)^{\bar{c}}V(t_m^e).
\end{equation}
During an attack interval $[t_m^s,t_m^e)\cap[0,t_c^F)$, \eqref{eq:time_scaling} makes $\gamma$ constant, so $\mu=\gamma^2$ in \eqref{eq:consensus_lyapunov}. Under the blackout branch of \eqref{eq:control}, the sliding variable satisfies $\dot{\mathbf{S}}=(\gamma-k_v)\mathbf{Z}_2+\mathcal{L}(\mathbf{d}-\eta\sign(\mathbf{v}))$. Differentiating \eqref{eq:consensus_lyapunov} now yields
\begin{align}\label{eq:Vdot_pre2}
    \dot{V} =& (\gamma-k_v)\mathbf{S}^{\top}\Xi\mathbf{S}
    +k_v\gamma\mathbf{S}^{\top}\Xi\mathbf{Z}_1\nonumber\\
    &-\gamma^3\mathbf{Z}_1^{\top}\Xi\mathbf{Z}_1
    +\mathbf{S}^{\top}\Xi\mathcal{L}(\mathbf{d}-\eta\sign(\mathbf{v})).
\end{align}
Young's inequality bounds the mixed term by
$k_v\gamma\mathbf{S}^{\top}\Xi\mathbf{Z}_1\leq(k_v/2)\mathbf{S}^{\top}\Xi\mathbf{S}+(k_v\gamma^2/2)\mathbf{Z}_1^{\top}\Xi\mathbf{Z}_1$.
For $\boldsymbol{\Phi}:=\Xi\mathcal{L}(\mathbf{d}-\eta\sign(\mathbf{v}))$, the disturbance bound implies $\|\boldsymbol{\Phi}\|_2\leq\bar{\omega}$, where $\bar{\omega}:=\|\mathcal{L}\|_2\lambda_{\max}(\Xi)\sqrt{N}(\bar{d}+\eta)$.
Applying $\mathbf{S}^{\top}\boldsymbol{\Phi}\leq(\|\mathbf{S}\|_2^2+\|\boldsymbol{\Phi}\|_2^2)/2$ to \eqref{eq:Vdot_pre2} and using \eqref{eq:consensus_lyapunov} establishes
\begin{equation}\label{eq:V_attack_growth}
    \dot{V}\leq\theta(t)V+\hat{\theta},
\end{equation}
with $\theta(t):=2|\gamma-k_v/2|+(2/\gamma^2)\max\{0, (k_v^2/2)-\gamma^3\}+1/\lambda_{\min}(\Xi)$ and $\hat{\theta}:=\bar{\omega}^2/2$.

The deadline bound \eqref{eq:realistic_settling_bound} and the attack-frequency bound in \Cref{assum:duration} ensure that the number $M$ of attacks starting before $t_c^F$ is finite, with $M\leq\mathcal{N}_A(0,t_c^F)\leq t_c^F/\nu_3+\nu_4$.
By the first-attainment definition \eqref{eq:realistic_settling_time}, every such attack ends before $t_c^F$: the remaining attack-free duration is positive at the attack onset and remains constant throughout the blackout.
Hence $\gamma(t_m^s)$ is finite for $m=1,\ldots,M$, and a finite constant $\bar{\theta}>0$ bounds $\theta$ on all these attack intervals. Take $\bar{\theta}:=\max_{1\leq m\leq M}\theta(t_m^s)$ when $M\geq1$, and $\bar{\theta}:=1$ when $M=0$.
Integration of \eqref{eq:V_attack_growth} over the $m$-th attack leads to the endpoint bound
\begin{equation}\label{eq:V_attack_endpoint}
    V(t_m^{e-})\leq e^{\bar{\theta}(t_m^e-t_m^s)}V(t_m^s)+\Delta_m,
\end{equation}
where $\Delta_m:=(\hat{\theta}/\bar{\theta})[e^{\bar{\theta}(t_m^e-t_m^s)}-1]\geq0$.

At switching instants before $t_c^F$, the plant states $\mathbf{x}$ and $\mathbf{v}$ and the scaling function $\gamma$ are continuous, so \eqref{eq:sliding_variable} ensures continuity of $\mathbf{S}$.
The coefficient $\mu$ equals $(c-1)\gamma^2/c$ on attack-free intervals and $\gamma^2$ during attacks. Across an attack onset, $\mu$ increases by at most the factor $c/(c-1)$; at an attack end, $\mu$ cannot increase.
Thus \eqref{eq:consensus_lyapunov} implies $V(t_m^s)\leq\alpha V(t_m^{s-})$ and $V(t_m^e)\leq V(t_m^{e-})\leq\alpha V(t_m^{e-})$, where $\alpha:=c/(c-1)$. If the first attack starts at $t=0$, take $V(t_1^{s-}):=V(0)$ in the initial switching estimate.
Applying the bounds on $V(t_m^s)$ and $V(t_m^e)$ to \eqref{eq:V_attack_endpoint} and using \eqref{eq:V_healthy1} yields
\begin{align}\label{eq:V_healthy2}
    V(t)\leq&\left(\frac{t_c^{m+1}-t}{t_c^{m+1}-t_m^e}\right)^{\bar{c}}\left[\alpha^2 e^{\bar{\theta}(t_m^e-t_m^s)}V(t_m^{s-})+\alpha\Delta_m\right].
\end{align}
Since $t-|\Upsilon_A(0,t)|$ is constant during $[t_m^s,t_m^e)$, the remaining attack-free duration has the same value at attack onset and recovery:
\begin{equation}\label{eq:clock_endpoint_identity}
    t_c^{m+1}-t_m^e=t_c^m-t_m^s>0.
\end{equation}
To propagate \eqref{eq:V_healthy2} through successive intervals, set $\Lambda_0:=V(0)/t_c^{\bar{c}}$ and define
\begin{equation}\label{eq:V_healthy3}
    \Lambda_m:=\alpha^2 e^{\bar{\theta}(t_m^e-t_m^s)}\Lambda_{m-1}
    +\frac{\alpha\Delta_m}{(t_c^{m+1}-t_m^e)^{\bar{c}}}.
\end{equation}
Suppose $V(t)\leq\Lambda_{m-1}(t_c^m-t)^{\bar{c}}$ for $t\in[t_{m-1}^e,t_m^s)$.
Taking the limit as $t\to(t_m^s)^-$ yields $V(t_m^{s-})\leq\Lambda_{m-1}(t_c^m-t_m^s)^{\bar{c}}$.
Substituting the bound on $V(t_m^{s-})$ into \eqref{eq:V_healthy2} and applying \eqref{eq:clock_endpoint_identity} establishes $V(t)\leq\Lambda_m(t_c^{m+1}-t)^{\bar{c}}$ on $[t_m^e,t_{m+1}^s)\cap[0,t_c^F)$, with $\Lambda_m$ defined by \eqref{eq:V_healthy3}.
Equation \eqref{eq:V_healthy1} with $m=0$ establishes $V(t)\leq\Lambda_0(t_c-t)^{\bar{c}}$ on the initial healthy interval. Induction using \eqref{eq:V_healthy3} extends the bound on $V(t)$ to every healthy interval before $t_c^F$.
Expanding the recursion \eqref{eq:V_healthy3} expresses $\Lambda_m$ in terms of the attack history:
\begin{align}\label{eq:V_healthy4}
    \Lambda_m=&\frac{\alpha^{2m}e^{\bar{\theta}\sum_{r=1}^{m}(t_r^e-t_r^s)}}{t_c^{\bar{c}}}V(0)\nonumber\\
    &+\sum_{q=1}^{m}\frac{\alpha^{2(m-q)+1}e^{\bar{\theta}\sum_{r=q+1}^{m}(t_r^e-t_r^s)}}{(t_c^{q+1}-t_q^e)^{\bar{c}}}\Delta_q.
\end{align}
Since $t_c>0$ and \eqref{eq:clock_endpoint_identity} ensures $t_c^{q+1}-t_q^e>0$, the finite sum in \eqref{eq:V_healthy4} defines finite coefficients $\Lambda_m$ for $m=0,\ldots,M$.
With $\Lambda:=\max_{0\leq m\leq M}\Lambda_m<\infty$, the healthy-interval estimates become
\begin{equation}\label{eq:V_healthy5}
    V(t)\leq\Lambda(t_c^{m+1}-t)^{\bar{c}}.
\end{equation}
The constant $\Lambda$ may depend on the initial state and attack history; the deadline bound \eqref{eq:realistic_settling_bound} is independent of the initial state.

There is a final healthy interval $[t_M^e,t_c^F)$, with $t_M^e=0$ if $M=0$, and \eqref{eq:realistic_settling_time} implies $t_c^{M+1}=t_c^F$.
Set $r(t):=t_c^F-t$ and $C:=\sqrt{2\Lambda/\lambda_{\min}(\Xi)}$.
Since $\mu=c(c-1)/r^2$ on this interval, the two nonnegative terms in \eqref{eq:consensus_lyapunov}, together with \eqref{eq:V_healthy5}, establish
\begin{align}\label{eq:terminal_disagreement_rates}
    \|\mathbf{S}(t)\|_2\leq& C r(t)^{\bar{c}/2},~
    \|\mathbf{Z}_1(t)\|_2\leq\frac{C}{\sqrt{c(c-1)}}r(t)^{\bar{c}/2+1}.
\end{align}
Using $\mathbf{Z}_2=\mathbf{S}-\gamma\mathbf{Z}_1$ from \eqref{eq:sliding_variable} and $\gamma=c/r$, \eqref{eq:terminal_disagreement_rates} yields
\begin{equation}\label{eq:terminal_velocity_rate}
    \|\mathbf{Z}_2(t)\|_2\leq C\left(1+\sqrt{\frac{c}{c-1}}\right)r(t)^{\bar{c}/2}.
\end{equation}
Both $\mathbf{Z}_1$ and $\mathbf{Z}_2$ tend to zero as $t\to(t_c^F)^-$.
Because $\ker(\mathcal{L})=\operatorname{span}\{\mathbf{1}_N\}$ by \Cref{lem:directed_graph}, convergence of $\mathbf{Z}_1$ and $\mathbf{Z}_2$ to zero establishes the pairwise position and velocity consensus limits \eqref{eq:consensus}.
Finally, \eqref{eq:terminal_disagreement_rates} bounds the attack-free consensus correction by
\begin{equation}\label{eq:terminal_correction_bound}
    k\gamma(t)\|\mathbf{S}(t)\|_2\leq ckC r(t)^{(\bar{c}-2)/2}.
\end{equation}
The correction remains bounded since $\bar{c}\geq2$.
When $c>2$, the strict gain condition \eqref{eq:consensus_gains} ensures $\hat{k}>1/c$ and $(c-1)/c>1/c$, so $\bar{c}>2$. The upper bound in \eqref{eq:terminal_correction_bound} then vanishes as $t\to(t_c^F)^-$, establishing $k\gamma(t)\mathbf{S}(t)\to\mathbf{0}_N$. As $t\to t_F^c$, $V(t)\to0$, by continuity of $V$ and since $\dot V\leq 0$, the consensus manifold remains invariant for all $t\ge t_F^{c-}$.
\end{proof}
In the absence of DoS attacks, \eqref{eq:realistic_settling_time} reduces to $t_c^F=t_c$. Hence \Cref{thm:consensus} establishes the position and velocity agreement limits \eqref{eq:consensus} at  $t_c$, despite bounded exogenous disturbances.
\begin{remark}
For a connected undirected graph, the normalization in \Cref{lem:directed_graph} yields $\boldsymbol{\xi}=\mathbf{1}_N/N$, $\Xi=\mathbf{I}_N/N$, and $\hat{\mathcal{L}}=\mathcal{L}/N$.
Hence $a(\hat{\mathcal{L}})=\lambda_2(\mathcal{L})/N$, and \eqref{eq:consensus_gains} reduces to $k>(1+1/c)/\lambda_2(\mathcal{L})$, with $c\geq2$, $k_v\geq0$, and $\eta\geq\bar{d}$ unchanged.
The discontinuous consensus term simplifies to $\eta\,\sign(\mathcal{L}\mathbf{S})$ inside \eqref{eq:control}.
\end{remark}
During communication blackouts, the local branch of \eqref{eq:control} regulates each agent's velocity through damping and disturbance compensation. Substitution into \eqref{eq:agent_dynamics} results in $\dot{v}_i=-k_v v_i-\eta\sign(v_i)+d_i$.
Under \eqref{eq:consensus_gains}, the velocity satisfies $v_i\dot{v}_i\leq-k_v v_i^2$ almost everywhere, so the speed of each agent remains nonincreasing throughout a blackout. The scaling function $\gamma$ remains fixed during each outage by \eqref{eq:time_scaling}; relative-position feedback and the evolution of $\gamma$ resume when communication is restored. The design of analogous local blackout feedback for higher-order agents remains a direction for future work.
\begin{remark}
The estimate \eqref{eq:terminal_correction_bound} ensures that the consensus correction $k\gamma(t)\mathbf{S}(t)$ remains bounded as $t\to(t_c^F)^-$ and vanishes when $c>2$. During attack-free intervals, the full control input in \eqref{eq:control} satisfies $    \|\mathbf{u}(t)\|_2\leq\|\mathbf{B}^{-1}\|_2
    \left(\|\mathbf{F}\|_2+k\gamma(t)\|\mathbf{S}(t)\|_2+\eta\sqrt{N}\right).$ During an attack interval, the local feedback branch instead implies $\|\mathbf{u}(t)\|_2\leq\|\mathbf{B}^{-1}\|_2\left(\|\mathbf{F}\|_2+k_v\|\mathbf{v}(t)\|_2+\eta\sqrt{N}\right)$.
Thus the full input is bounded on $[0,t_c^F)$ if $\mathbf{F}$, $\mathbf{B}^{-1}$, and the blackout velocities remain bounded.
\end{remark}
Implementing \eqref{eq:control} with finite actuation capacity requires choosing $t_c$ so that the control demand remains within the available input range. Although the convergence guarantee in \Cref{thm:consensus} holds for all finite initial states, the required control magnitude depends on the initial conditions. Actuator limits consequently restrict the range of implementable deadlines, even in the absence of DoS attacks.

\section{Simulations}
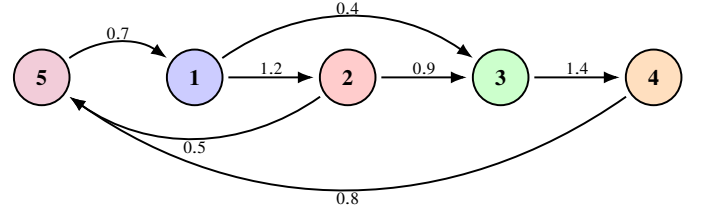
\begin{figure}[ht!]
    \centering
    \centering
        \resizebox{\linewidth}{!}{%
        \begin{tikzpicture}[
    >=Latex,
    vertex/.style={
        circle,
        draw=black,
        thick,
        minimum size=8mm,
        inner sep=0pt,
        font=\small\bfseries
    },
    edge/.style={
        ->,
        thick,
        shorten >=2pt,
        shorten <=2pt
    },
    weight/.style={
        fill=none,
        inner sep=1pt,
        font=\scriptsize
    }
]

\node[vertex, fill=purple!20] (v5) at (0,0)   {5};
\node[vertex, fill=blue!20]   (v1) at (2.2,0) {1};
\node[vertex, fill=red!20]    (v2) at (4.4,0) {2};
\node[vertex, fill=green!20]  (v3) at (6.6,0) {3};
\node[vertex, fill=orange!25] (v4) at (8.8,0) {4};

\draw[edge] (v5) to[bend left=35] node[weight, above] {$0.7$} (v1);
\draw[edge] (v1) -- node[weight, above] {$1.2$} (v2);
\draw[edge] (v2) -- node[weight, above] {$0.9$} (v3);
\draw[edge] (v3) -- node[weight, above] {$1.4$} (v4);

\draw[edge] (v1) to[bend left=35] node[weight, above] {$0.4$} (v3);
\draw[edge] (v2) to[bend left=35] node[weight, below] {$0.5$} (v5);
\draw[edge] (v4) to[bend left=35] node[weight, below] {$0.8$} (v5);

\end{tikzpicture}%
        }

        \caption{Interaction Topology ($\mathcal{G}$).}
        \label{fig:strongly_connected1}
\end{figure}    
To demonstrate the efficacy of the proposed strategy, we consider $5$ agents interacting over a weighted, strongly connected digraph $\mathcal{G}$. Each agent is modeled as an ideal single-link robotic manipulator with $f_i=-\frac{g}{l}\sin x_i-\frac{\bar\nu}{ml^2}v_i$ and $b_i=\frac{1}{ml^2}$, where $x_i$ denotes the angular position, $m=2~\mathrm{kg}$ is the link mass, $l=1~\mathrm{m}$ is the link length, $\bar\nu=6~\mathrm{Nms}$ is the friction coefficient, $\tau_i$ is the applied torque (control input), and $g=9.81~\mathrm{m/s^2}$ is the gravitational acceleration. The external disturbance is modeled as $d_i=0.5i\sin((0.5+0.1i)t)+0.05i\cos((1+0.2i)t)$, $i=1,\dots,5$.The initial conditions are set to $x(0) = [-2.5,\ 1.25,\ 3.75,\ -1.25,\ 5.0]^\top$ and $v(0) = [1.0,\ -0.6,\ 0.4,\ 0.0,\ -0.8]^\top$. The controller gains are set to $c=3$, $k=5$, and $\eta=3.5$, with convergence time $t_c=5~\mathrm{s}$. The discontinuous $\operatorname{sign}(\cdot)$ in \eqref{eq:control} is replaced by $\tanh(\cdot/\varepsilon)$, where $\varepsilon>0$ is small, to mitigate numerical chattering. We first consider the case with no DoS attacks. The corresponding results are illustrated in \Cref{fig:case_1}. One can observe from \Cref{fig:s1} that $S$ converges to zero within $t=0.5~\mathrm{s}$ and remains there despite the disturbances. On the sliding manifold, $\mathbf{Z}_2=\dot{\mathbf{Z}}_1\approx-\gamma(t)\mathbf{Z}_1$, which guarantees consensus of $x_i$ and $v_i$ as $t\to t_c^-$, and consensus is preserved thereafter, as can be observed from \Cref{fig:x1} and \Cref{fig:v1}. The corresponding control input  (shown in \Cref{fig:s1}), remains bounded throughout.
\begin{figure*}[h!]
    \centering
    \begin{subfigure}[t]{0.32\linewidth}
        \centering
        \includegraphics[width=\linewidth]{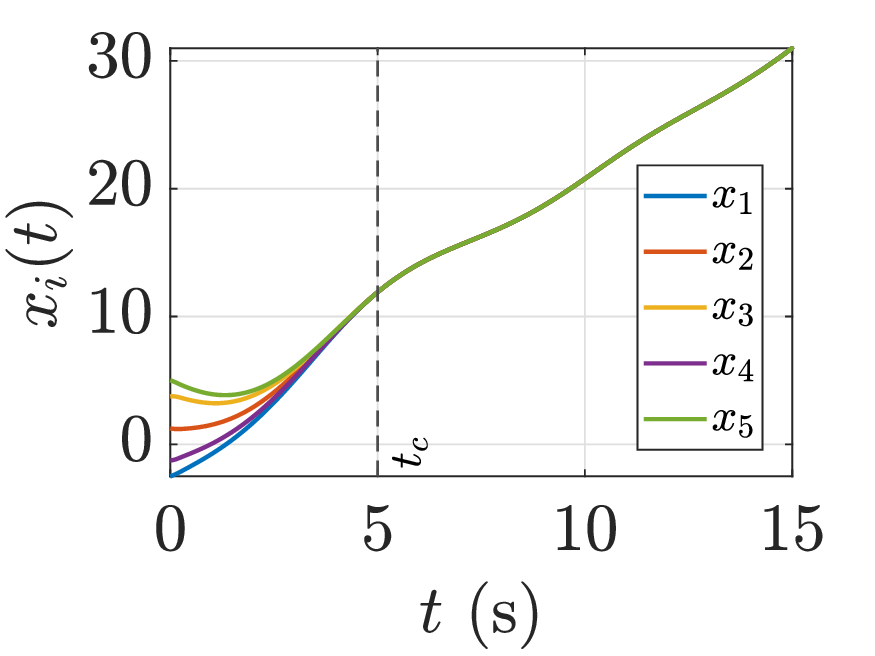}
        \caption{$x_i(t)$.}
        \label{fig:x1}
    \end{subfigure}
    \begin{subfigure}[t]{0.32\linewidth}
        \centering
        \includegraphics[width=\linewidth]{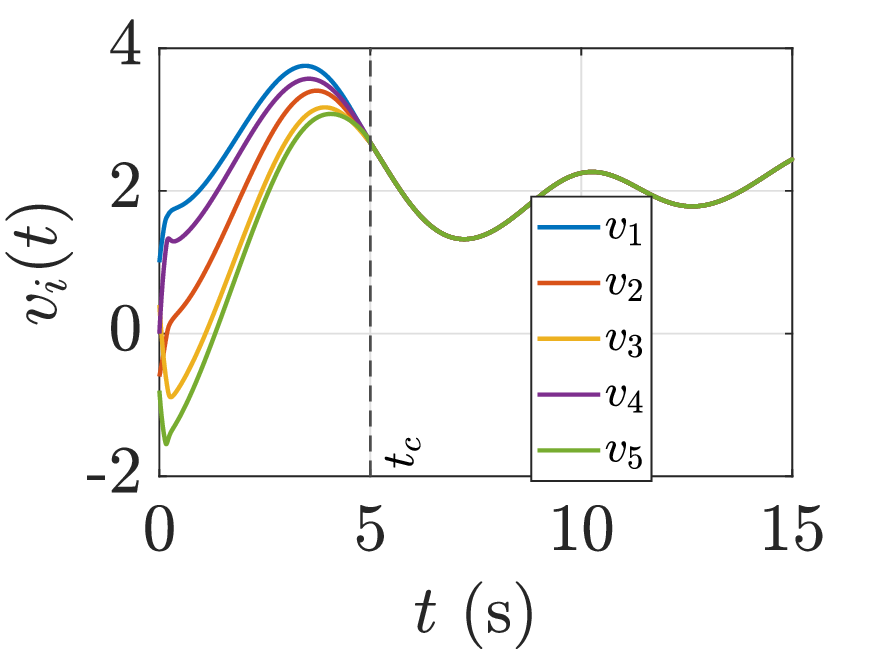}
        \caption{$v_i(t)$.}
        \label{fig:v1}
    \end{subfigure}
    \begin{subfigure}[t]{0.32\linewidth}
        \centering
        \includegraphics[width=\linewidth]{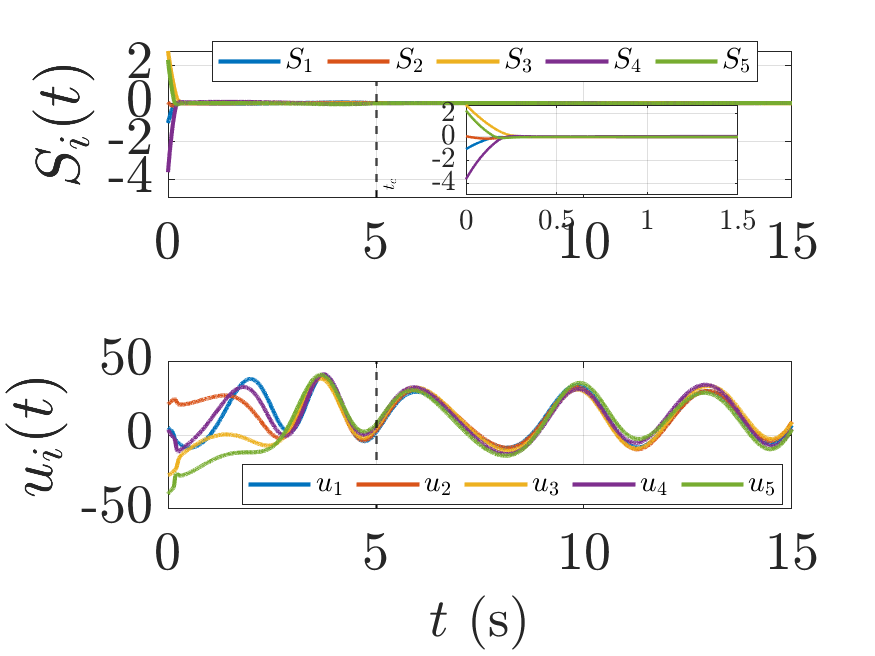}
        \caption{$S(t)$ and $u_i(t)$}
        \label{fig:s1}
    \end{subfigure}
    \caption{Prescribed time consensus under no DoS attacks.}
    \label{fig:case_1}
\end{figure*}

For the next case, we consider the system under DoS attacks in the absence of disturbances, with the gains $k_v$ and $\eta$ set to zero. The corresponding results are presented in \Cref{fig:case_2}, where the DoS attack intervals are indicated by the gray bands. In \Cref{fig:u2}, the quantity $t_{\rm go}(t)=t_c+|\Upsilon_A(0,t)|-t$ is depicted, indicating a realistic settling time of $t_c^F=6.5,\mathrm{s}$, at which $t_{\rm go}(t)$ reaches zero. As shown in \Cref{fig:v2}, $\dot{v}_i=0$ during the DoS intervals. Nevertheless, consensus in both velocity and position is achieved by $t_c^F$, as illustrated in \Cref{fig:v2} and \Cref{fig:x2} respectively. Furthermore, one can note that the consensus condition remains invariant under subsequent DoS attacks. Corresponding control inputs are shown in \Cref{fig:u2}. It can also be observed that abrupt changes occur at the transitions between DoS and healthy communication intervals; a smooth transition mechanism will be considered in future work.
 \begin{figure*}[h!]
    \centering
    \begin{subfigure}[t]{0.32\linewidth}
        \centering
        \includegraphics[width=\linewidth]{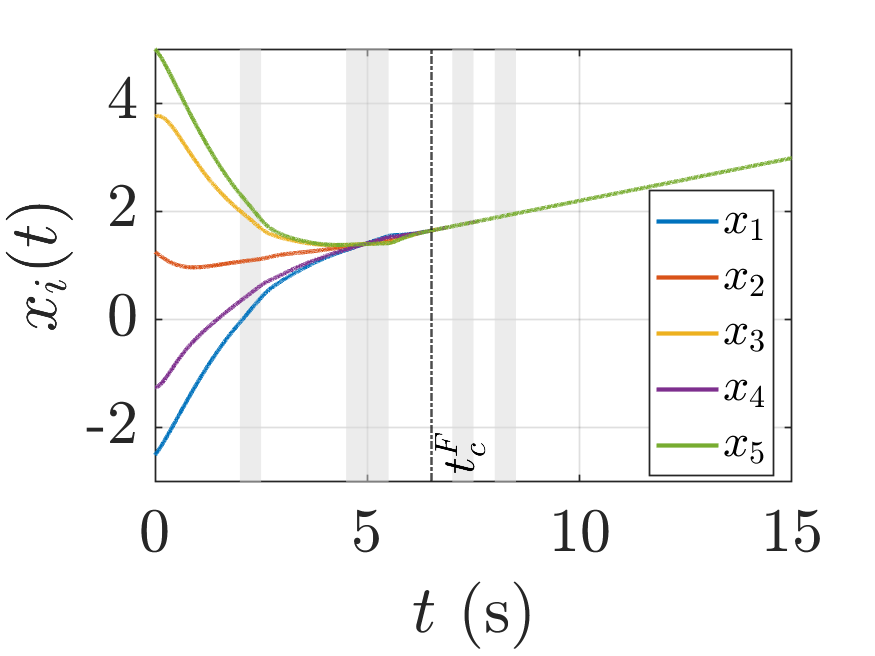}
        \caption{$x_i(t)$.}
        \label{fig:x2}
    \end{subfigure}
    \begin{subfigure}[t]{0.32\linewidth}
        \centering
        \includegraphics[width=\linewidth]{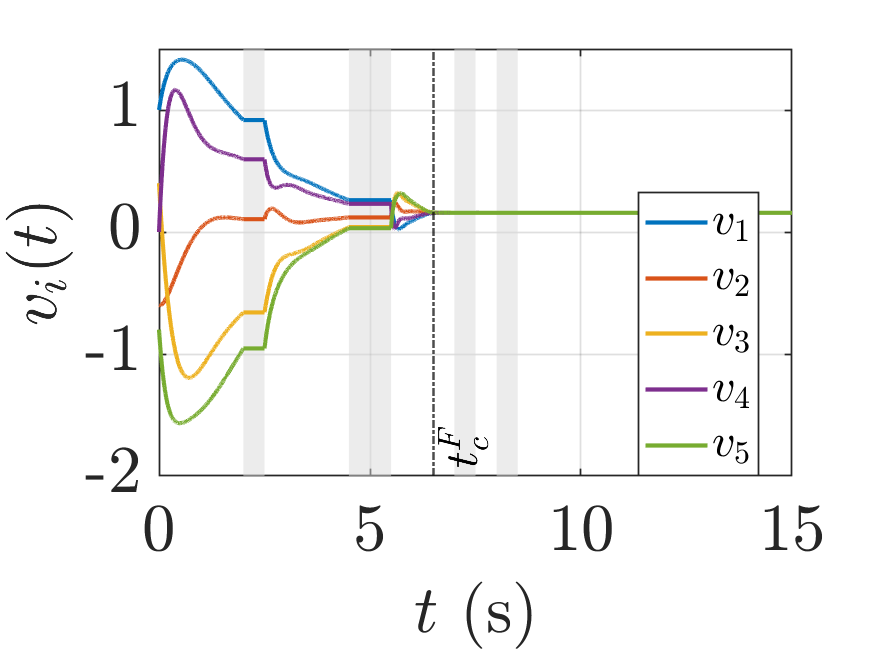}
        \caption{$v_i(t)$.}
        \label{fig:v2}
    \end{subfigure}
    \begin{subfigure}[t]{0.32\linewidth}
        \centering
        \includegraphics[width=\linewidth]{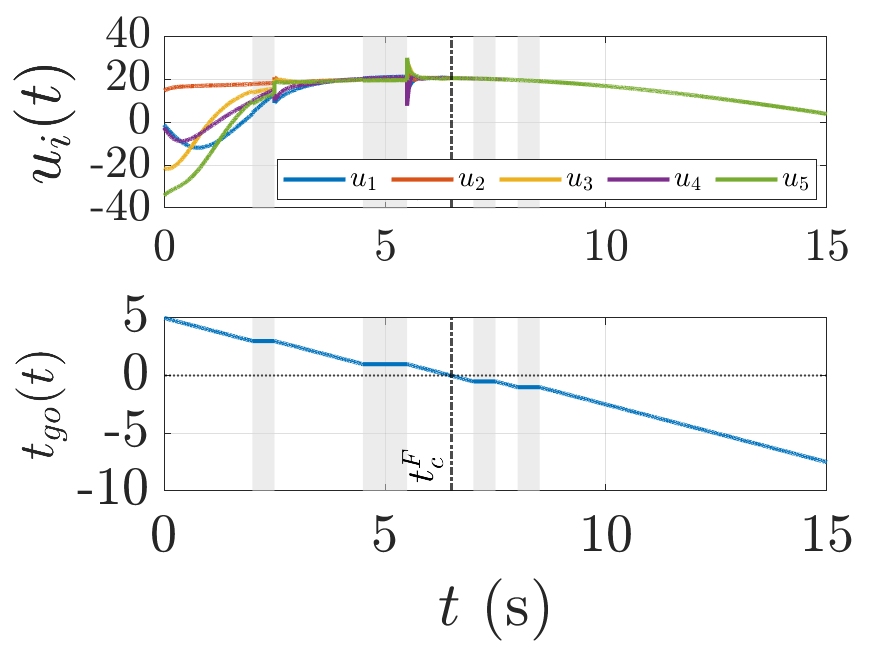}
        \caption{$u_i(t)$ and $t_{\rm go}(t)$}
        \label{fig:u2}
    \end{subfigure}
    \caption{Consensus under DoS attacks with no disturbances.}
    \label{fig:case_2}
\end{figure*}
 \begin{figure*}[h!]
    \centering
    \begin{subfigure}[t]{0.32\linewidth}
        \centering
        \includegraphics[width=\linewidth]{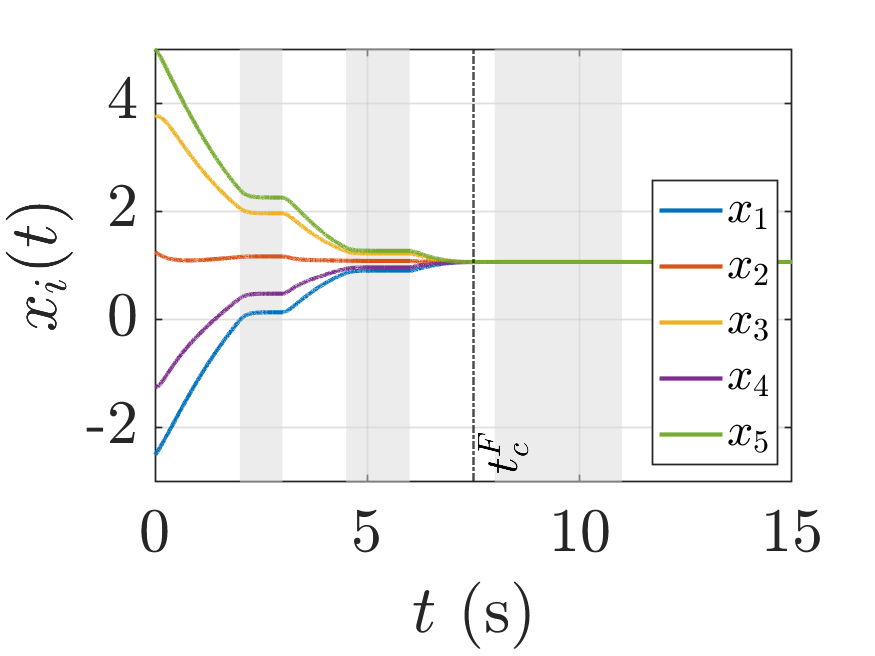}
        \caption{$x_i(t)$.}
        \label{fig:x3}
    \end{subfigure}
    \begin{subfigure}[t]{0.32\linewidth}
        \centering
        \includegraphics[width=\linewidth]{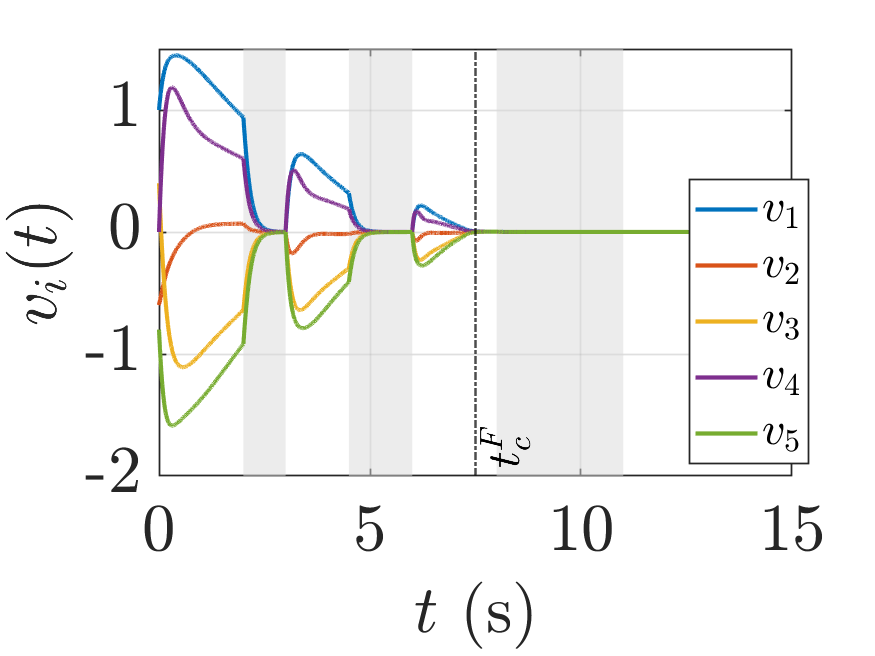}
        \caption{$v_i(t)$.}
        \label{fig:v3}
    \end{subfigure}
    \begin{subfigure}[t]{0.32\linewidth}
        \centering
        \includegraphics[width=\linewidth]{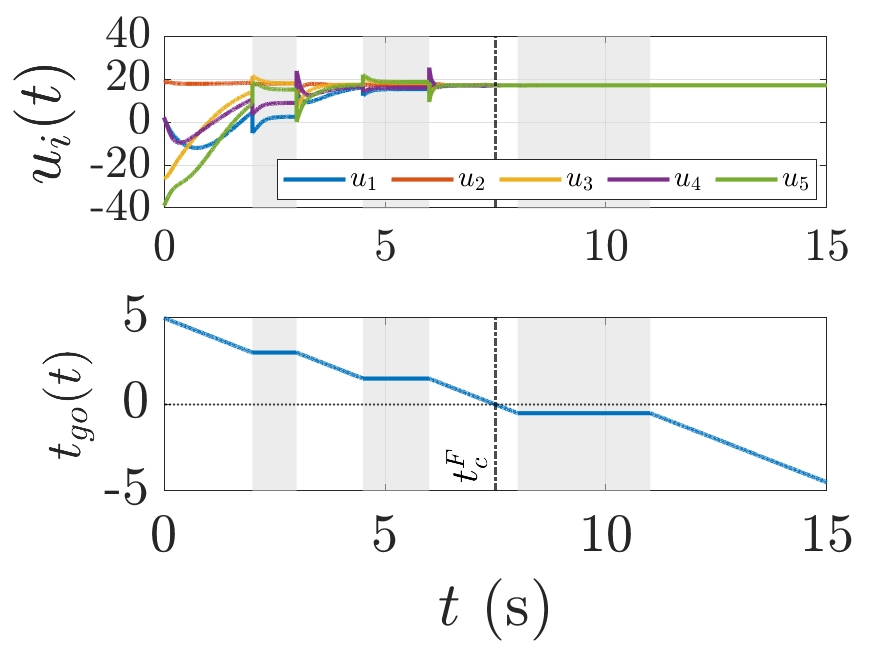}
        \caption{$u_i(t)$ and $t_{\rm go}(t)$}
        \label{fig:u3}
    \end{subfigure}
    \caption{Consensus under DoS attacks with disturbances.}
    \label{fig:case_3}
\end{figure*}

Next, we consider the system under DoS attacks in the presence of disturbances, with $k_v=3$ and $\eta=3.5$. A different DoS attack sequence is considered, and the corresponding results are presented in \Cref{fig:case_3}. From $t_{\rm go}(t)$ in \Cref{fig:u3}, one can note that $t_c^F=7.5~\rm s$. During the DoS intervals, $v_i$ decreases toward zero, while velocity and position consensus are achieved by $t_c^F$, as shown in \Cref{fig:v3} and \Cref{fig:x3}, respectively. The corresponding control inputs in \Cref{fig:u3} remain bounded, and the consensus remains invariant under subsequent DoS attacks and disturbances.

\section{Conclusions}\label{sec:conclusions}
We developed a switching sliding-mode protocol for predefined-time leaderless consensus of nonlinear second-order agents under bounded disturbances and DoS attacks. The proposed design uses mutual position and velocity disagreement over a strongly connected directed graph, avoids inversion of the singular graph Laplacian, and regulates motion through local velocity feedback during complete communication blackouts. Under bounds on attack duration and frequency, we established convergence by a realistic settling time that accounts for accumulated outages while retaining a prescribed attack-free duration independent of the initial conditions. We derived bounds that track disagreement amplification and disturbance accumulation during blackouts and quantify disagreement decay after communication recovery. The resulting gain conditions ensure position and velocity agreement as the realistic settling time is approached, while keeping the consensus correction bounded despite the singular time scaling. Future work will address the selection of feasible convergence durations under actuator limits and the extension of local blackout feedback to higher-order agents.

\bibliographystyle{IEEEtran}
\bibliography{references_resilient}

\begin{thebibliography}{10}
\providecommand{\url}[1]{#1}
\csname url@samestyle\endcsname
\providecommand{\newblock}{\relax}
\providecommand{\bibinfo}[2]{#2}
\providecommand{\BIBentrySTDinterwordspacing}{\spaceskip=0pt\relax}
\providecommand{\BIBentryALTinterwordstretchfactor}{4}
\providecommand{\BIBentryALTinterwordspacing}{\spaceskip=\fontdimen2\font plus
\BIBentryALTinterwordstretchfactor\fontdimen3\font minus
  \fontdimen4\font\relax}
\providecommand{\BIBforeignlanguage}[2]{{%
\expandafter\ifx\csname l@#1\endcsname\relax
\typeout{** WARNING: IEEEtran.bst: No hyphenation pattern has been}%
\typeout{** loaded for the language `#1'. Using the pattern for}%
\typeout{** the default language instead.}%
\else
\language=\csname l@#1\endcsname
\fi
#2}}
\providecommand{\BIBdecl}{\relax}
\BIBdecl

\bibitem{9448334}
Y.~Bai, Y.~Wang, M.~Svinin, E.~Magid, and R.~Sun, ``Adaptive multi-agent
  coverage control with obstacle avoidance,'' \emph{IEEE Control Systems
  Letters}, vol.~6, pp. 944--949, 2022.

\bibitem{9924233}
A.~Sinha and Y.~Cao, ``3-d nonlinear guidance law for target
  circumnavigation,'' \emph{IEEE Control Systems Letters}, vol.~7, pp.
  655--660, 2023.

\bibitem{9274339}
A.~Sinha, S.~R. Kumar, and D.~Mukherjee, ``Cooperative salvo based active
  aircraft defense using impact time guidance,'' \emph{IEEE Control Systems
  Letters}, vol.~5, no.~5, pp. 1573--1578, 2021.

\bibitem{9000526}
A.~Sinha and S.~R. Kumar, ``Supertwisting control-based cooperative salvo
  guidance using leader–follower approach,'' \emph{IEEE Transactions on
  Aerospace and Electronic Systems}, vol.~56, no.~5, pp. 3556--3565, 2020.

\bibitem{yu2009second}
W.~Yu, G.~Chen, M.~Cao, and J.~Kurths, ``Second-order consensus for multiagent
  systems with directed topologies and nonlinear dynamics,'' \emph{IEEE
  Transactions on Systems, Man, and Cybernetics, Part B (Cybernetics)},
  vol.~40, no.~3, pp. 881--891, 2009.

\bibitem{li2011finite}
S.~Li, H.~Du, and X.~Lin, ``Finite-time consensus algorithm for multi-agent
  systems with double-integrator dynamics,'' \emph{Automatica}, vol.~47, no.~8,
  pp. 1706--1712, 2011.

\bibitem{hong2018novel}
H.~Hong, W.~Yu, J.~Fu, and X.~Yu, ``A novel class of distributed fixed-time
  consensus protocols for second-order nonlinear and disturbed multi-agent
  systems,'' \emph{IEEE Transactions on Network Science and Engineering},
  vol.~6, no.~4, pp. 760--772, 2018.

\bibitem{ning2022fixed}
B.~Ning, Q.-L. Han, Z.~Zuo, L.~Ding, Q.~Lu, and X.~Ge, ``Fixed-time and
  prescribed-time consensus control of multiagent systems and its applications:
  A survey of recent trends and methodologies,'' \emph{IEEE Transactions on
  Industrial Informatics}, vol.~19, no.~2, pp. 1121--1135, 2022.

\bibitem{wang2018prescribed}
Y.~Wang, Y.~Song, D.~J. Hill, and M.~Krstic, ``Prescribed-time consensus and
  containment control of networked multiagent systems,'' \emph{IEEE
  transactions on cybernetics}, vol.~49, no.~4, pp. 1138--1147, 2018.

\bibitem{ding2023prescribed}
T.-F. Ding, M.-F. Ge, C.~Xiong, Z.-W. Liu, and G.~Ling, ``Prescribed-time
  formation tracking of second-order multi-agent networks with directed
  graphs,'' \emph{Automatica}, vol. 152, p. 110997, 2023.

\bibitem{meng2026practical}
M.~Meng, L.~Li, Q.~Kang, and Q.~Gan, ``Practical prescribed-time scaled
  consensus of second-order multi-agent systems via event-triggered and
  time-based generator approaches,'' \emph{Communications in Nonlinear Science
  and Numerical Simulation}, p. 110608, 2026.

\bibitem{wang2018sliding}
G.~Wang, X.~Wang, and S.~Li, ``Sliding-mode consensus algorithms for disturbed
  second-order multi-agent systems,'' \emph{Journal of the Franklin Institute},
  vol. 355, no.~15, pp. 7443--7465, 2018.

\bibitem{wei2025further}
P.~Wei, M.~Luo, J.~Cheng, and X.~Wang, ``Further results on coded-based
  predefined-time consensus via nonsingular sliding mode control for multiple
  aerial vehicles,'' \emph{Information Sciences}, vol. 721, p. 122637, 2025.

\bibitem{shang2021event}
Y.~Shang, C.-L. Liu, and K.-C. Cao, ``Event-triggered consensus control of
  second-order nonlinear multi-agent systems under denial-of-service attacks,''
  \emph{Transactions of the Institute of Measurement and Control}, vol.~43,
  no.~10, pp. 2272--2281, 2021.

\bibitem{liang2025fixed}
J.~Liang, Z.~Chen, Z.~Yu, and H.~Jiang, ``Fixed-time consensus of second-order
  multi-agent systems based on event-triggered mechanism under dos attacks,''
  \emph{AIMS Mathematics}, vol.~10, no.~1, pp. 1501--1528, 2025.

\bibitem{li2026event}
B.~Li, Z.~Wang, W.~Wang, Q.~Gao, and J.~L{\"u}, ``Event-/self-triggered
  communication for dos-resilient consensus in multiagent systems with
  application to leo satellite formation,'' \emph{IEEE Transactions on
  Industrial Informatics}, 2026.

\bibitem{li2026dual}
B.~Li, Q.~Gao, Z.~Wang, W.~Wang, and J.~L{\"u}, ``Dual-terminal dynamic
  event-triggering leader-following and leaderless consensus for multi-agent
  systems under distributed dos attack,'' \emph{International Journal of Robust
  and Nonlinear Control}, 2026.

\bibitem{ye2026prescribed}
Z.~Ye, D.~Zhang, L.~Liu, Y.~Shi, and G.~Feng, ``Prescribed time consensus of
  leader-following multi-agent systems under dos attacks,'' \emph{Automatica},
  vol. 188, p. 112953, 2026.

\bibitem{zhu2026prescribed}
W.~Zhu, Y.~Luo, J.~Cao, D.~Yue, and W.~Xia, ``Prescribed-time event-triggered
  consensus control of nonlinear multi-agent systems under dos attacks and
  time-delays,'' \emph{Neurocomputing}, p. 133360, 2026.

\bibitem{sharma2026selftriggered}
M.~Sharma, H.~Shanmugam, and S.~R. Mahapatro, ``Self-triggered prescribed-time
  time varying formation control of second-order multi-agent systems with
  denial of service ({DoS}) attacks,'' \emph{Scientific Reports}, vol.~16, p.
  26672, 2026.

\bibitem{cao2026distributed}
W.~Cao, L.~Liu, D.~Zhang, and G.~Feng, ``Distributed resilient fixed-time
  control for cooperative output regulation of mass over directed graphs under
  dos attacks,'' \emph{IEEE Transactions on Automatic Control}, 2026.

\bibitem{yang2026finite}
H.~Yang and Y.~Wang, ``Finite-gain based prescribed-time consensus control for
  multi-agent systems under switching topology,'' \emph{Automatica}, vol. 183,
  p. 112686, 2026.

\end{thebibliography}
\end{document}